\documentclass[runningheads]{llncs}

\usepackage[T1]{fontenc}
\usepackage{graphicx}
\usepackage{amsmath}
\usepackage{amssymb}
\usepackage{algorithmic}
\usepackage{float}
\usepackage{color}
\usepackage{xcolor}
\usepackage{subcaption}
\usepackage{textcomp}
\usepackage{marvosym}
\usepackage{mathtools}
\newtheorem{assumption}{Assumption}
\usepackage{siunitx}
\usepackage{longtable}
\usepackage{booktabs}
\usepackage{array}
\usepackage[linesnumbered,ruled,vlined]{algorithm2e}

\usepackage[capitalize,noabbrev,nameinlink]{cleveref}
\crefformat{equation}{(#2#1#3)}
\Crefname{algocfline}{Algorithm}{Algorithms}
\Crefname{algocf}{line}{lines}
\Crefname{assumption}{Assumption}{Assumptions}
\crefrangeformat{assumption}{Assumptions~#3#1#4-#5#2#6}
\crefrangeformat{equation}{(#3#1#4)-(#5#2#6)}

\newcommand{\mcl}[1]{\mathcal{#1}}
\newcommand{\mbb}[1]{\mathbb{#1}}

\providecommand{\ubar}[1]{\underline{#1}}
\providecommand{\obar}[1]{\overline{#1}}

\DeclareMathOperator*{\minimize}{\mathrm{minimize}}
\DeclareMathOperator*{\argmin}{\mathrm{argmin}}
\DeclareMathOperator*{\subjectto}{\mathrm{subject~to}}

\newcommand{\cf}{{\it cf.}}
\newcommand{\eg}{{\it e.g.}}
\newcommand{\ie}{{\it i.e.}}

\providecommand{\BR}{\mathrm{BR}}
\providecommand{\saf}{\mathrm{s}}
\providecommand{\des}{\mathrm{des}}

\newcommand{\ours}{IBR-GCS}

\begin{document}

\title{Game-Theoretic Autonomous Driving: \\ A Graphs of Convex Sets Approach}

\author{Nikolaj Käfer\inst{1,2}$^*$ \and
Ahmed Khalil\inst{2}$^*$ \and
Edward Huynh \inst{2}$^*$ \and
Efstathios Bakolas\inst{2}
\and David Fridovich-Keil\inst{2}}

\authorrunning{N. Käfer et al.}

\institute{ETH Zurich, Zurich, ZH 8092, Switzerland \\
\email{nkaefer@ethz.ch} \and
The University of Texas at Austin, Austin, TX 78712, USA \\
\email{\{akhalil, edhuynh, dfk\}@utexas.edu, bakolas@austin.utexas.edu}
}

\makeatletter 
\def\thanks#1{\protected@xdef\@thanks{\@thanks \protect\footnotetext{#1}}}
\thanks{$^{*}$Equal contribution}
\makeatother

\maketitle 

\begin{abstract}

Multi-vehicle autonomous driving couples strategic interaction with hybrid (discrete-continuous) maneuver planning under shared safety constraints. We introduce \ours, an Iterative Best Response (IBR) planning approach based on the Graphs of Convex Sets (GCS) framework that models highway driving as a generalized noncooperative game. \ours\ integrates combinatorial maneuver reasoning, trajectory planning, and game-theoretic interaction within a unified framework. The key novelty is a vehicle-specific, strategy-dependent GCS construction. Specifically, at each best-response update, each vehicle builds its own graph conditioned on the current strategies of the other vehicles, with vertices representing lane-specific, time-varying, convex, collision-free regions and edges encoding dynamically feasible transitions. This yields a shortest-path problem in GCS for each best-response step, which admits an efficient convex relaxation that can be solved using convex optimization tools without exhaustive discrete tree search. We then apply an iterative best-response scheme in which vehicles update their trajectories sequentially and provide conditions under which the resulting inexact updates converge to an approximate generalized Nash equilibrium. Simulation results across multi-lane, multi-vehicle scenarios demonstrate that \ours\ produces safe trajectories and strategically consistent interactive behaviors.

\end{abstract}

\section{Introduction}

The deployment of self-driving cars promises improvements in safety, efficiency, and comfort in day-to-day transportation, but it also introduces fundamental challenges in planning with interactions. In mixed traffic, coordination is decentralized, and vehicle interactions combine cooperative and competitive elements: all vehicles must comply with shared traffic rules to ensure safety, yet strategic behavior naturally arises in scenarios such as merging, lane changing, and overtaking. To model the pursuit of individual objectives under shared constraints, noncooperative game theory has been widely used to cast autonomous driving as a complete-information game \cite{game_drive,game_lane,game_dec}. Compared to model predictive control approaches that ignore strategic interaction, game-theoretic formulations can better capture anticipatory behavior and improve performance in interactive scenarios \cite{out_MPC_game,Yan2023Decision}. A particularly useful subclass is (generalized) potential games, where unilateral incentives are aligned with a single scalar potential function, yielding algorithmic and theoretical tractability \cite{Bhatt_2025,pot_ilqr,pot_game_drive,pot_urb}. The versatility of potential games has led to their adoption in diverse multi-agent settings~\cite{bakolas2021TaskAssignment,lee2021RelayPursuit,best_response_drone,dist_potential}.

A key difficulty in multi-vehicle planning is that it is inherently hybrid: discrete decisions arise from lane selection and maneuver choices, while continuous decisions are governed by vehicle dynamics. This hybrid structure is often modeled via mixed-integer programs \cite{DeitsTedrake2015MIP,EarlDAndrea2002MILP,Mixed_Pot_game,RichardsHow2002MILP}, but such formulations can tightly couple discrete and continuous variables, limiting scalability and making it difficult to exploit problem structure in large multi-vehicle settings \cite{combin_drive,PadenCYYF16}.

To address these challenges, we leverage Graphs of Convex Sets (GCS) \cite{GCS_thesis,GCS}, which is a recently proposed optimization framework that has been successfully applied to trajectory optimization and motion planning problems with collision avoidance \cite{Kurtz2023GCSMotionPlanning,GCS_traj,GCS_SPP}. In GCS, discrete decisions are represented by a graph topology: vertices correspond to convex feasible regions ($\eg$, safe sets for each vehicle), and edges encode feasible transitions between them ($\eg$, collision-free trajectories from one feasible region to another). This representation preserves convexity within regions and along transitions, and it admits a convex relaxation that is often tight in practice \cite{GCS}, enabling efficient computation while capturing combinatorial structure.

In this work, we present \ours, a GCS-based method for multi-vehicle highway driving that models interaction as a noncooperative game with potentially conflicting objectives while respecting traffic rules. The proposed approach, \ours, integrates combinatorial maneuver reasoning, trajectory planning, and game-theoretic interaction within a unified framework. At each iteration, vehicle $i$ constructs a strategy-dependent GCS whose vertices represent lane-specific collision-free regions over time and whose edges encode dynamically feasible, safety-preserving transitions. Given the current strategies of the other vehicles, the resulting best-response computation reduces to a shortest-path problem (SPP) in GCS, yielding a mixed-integer convex formulation with an efficient convex relaxation that is often tight in practice. We embed these updates in an Iterative Best-Response (IBR) scheme and provide conditions under which the resulting inexact best-response dynamics converge to an approximate generalized Nash equilibrium (GNE). We evaluate \ours\ in numerical simulations of multi-lane highway scenarios, demonstrating safe trajectories and strategically consistent behaviors.

The remainder of the paper is organized as follows. Section~\ref{sec_preliminaries} reviews graphs of convex sets and generalized potential games. Section~\ref{sec_problem_formulation} formalizes the multi-vehicle driving game and states the key assumptions. Section~\ref{sec_approach} presents \ours, including the strategy-dependent GCS construction, the iterative best-response algorithm, and the accompanying error analysis. Section \ref{sec_simulation_results} presents simulation results. Finally, Section \ref{sec_conclusion} discusses conclusions.

\section{Preliminaries} \label{sec_preliminaries}

This section reviews GCS and generalized potential games, which form the basis for the proposed approach.

\subsection{Shortest Path Problem in Graphs of Convex Sets} \label{sec_gcs}

The GCS framework provides a graph-based representation for optimization problems that combine continuous decision variables with a combinatorial structure. The framework was originally introduced for trajectory optimization and motion planning with obstacle avoidance. A detailed formulation can be found in \cite{GCS_thesis,GCS}.

\begin{definition}[Graph of Convex Sets]
A graph of convex sets is a directed graph $\mcl{G}=(\mcl{V}, \mcl{E})$ with vertex set $\mcl{V}$ and edge set $\mcl{E}$, where each vertex $V \in \mcl{V}$ is associated with a convex set $\mcl{X}_V \subset \mbb{R}^n$, and each directed edge $e = (U,W) \in \mcl{E}$ represents an admissible one-step move from vertex $U$ to vertex $W$ ($\ie$, between the sets $\mcl{X}_U$ and $\mcl{X}_W$).
\end{definition}

In the relevant literature, vertices typically correspond to convex feasible regions that are subsets of the state or trajectory space, while edges encode admissible transitions between regions, such as continuity, dynamical feasibility, or safety constraints \cite{Kurtz2023GCSMotionPlanning,GCS_traj}.

Each vertex $V\in\mcl{V}$ is associated with a continuous decision variable $x_V \in \mcl{X}_V$. For each edge $e =(U,W) \in \mcl{E}$, convex constraints may couple $(x_U,  x_W)$ via $(x_U, x_W) \in \mcl{X}_e$, where $\mcl{X}_e \subset \mbb{R}^n \times \mbb{R}^n$ is a convex set. Each edge may also be assigned a convex nonnegative cost function $c_e: \mcl{X}_e \to \mbb{R}_{\ge 0}$. Additionally, vertices may be assigned convex nonnegative costs $c_V: \mcl{X}_V \to \mbb{R}_{\ge 0}$.

We recall the definition of a path in a graph of convex sets.

\begin{definition}[Path]
Given a source vertex $V_0$ and a target vertex $V_K \neq V_0$, a path is a sequence of vertices $(V_0, \dots, V_K)$ and directed edges such that $(V_{k-1}, V_k) \in \mcl{E}$ for all $k=1,\dots,K$.
\end{definition}

To find a minimum-cost path through $\mcl{G}$, we introduce binary decision variables $z_e \in \{0,1\}$ for all $e \in \mathcal{E}$ to indicate whether an edge $e$ is selected, and vertex-selection variables $y_V \in [0,1]$ for all $V \in \mathcal{V}$ (equal to $1$ for vertices $V$ on the selected path and $0$ otherwise under the flow constraints below). For each vertex $V \in \mcl{V}$, let $\mcl{E}^{\mathrm{in}}_V \coloneqq \{(U, V) \in \mcl{E}\}$ and $\mcl{E}^{\mathrm{out}}_V \coloneqq \{(V, W) \in \mcl{E}\}$ denote its sets of incoming and outgoing edges, respectively. A standard mixed-integer formulation of the GCS shortest path problem is as follows:
\begin{subequations} \label{eq_gcs_shortest_path}
\begin{align}
& \minimize_{\substack{\{x_V, y_V\}_{V \in \mcl{V}},\\ \{z_e\}_{e \in \mcl{E}},\\}}
&& \sum_{e=(U,W) \in \mcl{E}} z_e\, c_e(x_U, x_W) + \sum_{V\in\mcl{V}} y_V\, c_V(x_V) \label{eq_gcs_shortest_path_obj} \\
& \subjectto
&& \sum_{e \in \mcl{E}^{\mathrm{in}}_V} z_e + \delta_{V,V_{0}} = \sum_{e \in \mcl{E}^{\mathrm{out}}_V} z_e + \delta_{V,V_{K}}, && \forall V \in \mcl{V}, \label{eq_gcs_shortest_path_flow} \\
&&& y_V = \delta_{V,V_0} + \sum_{e\in\mcl{E}^{\mathrm{in}}_V} z_e, && \forall V\in\mcl{V}, \label{eq_gcs_shortest_path_ydef}\\
&&& x_V \in \mcl{X}_V, && \forall V \in \mcl{V}, \label{eq_gcs_shortest_path_feas} \\
&&& z_e = 1 \ \Rightarrow\ (x_U,x_W)\in\mcl{X}_e, && \forall e=(U,W)\in\mcl{E}, \label{eq_gcs_shortest_path_edgefeas} \\
&&& z_e \in \{0,1\}, && \forall e \in \mcl{E}, \label{eq_gcs_shortest_path_binary}\\
&&& 0 \le y_V \le 1, && \forall V\in\mcl{V}. \label{eq_gcs_shortest_path_ybox}
\end{align}
\end{subequations}
Here, $\delta_{V, W}$ denotes the Kronecker delta, where $\delta_{V, W}=1$, if $V = W$, and $\delta_{V, W}=0$, otherwise. In \eqref{eq_gcs_shortest_path_flow}, the terms $\delta_{V, V_0}$ and $\delta_{V, V_K}$ therefore create the required unit flow imbalance at the source vertex $V_0$ and target vertex $V_K$. In contrast to the classical SPP \cite{short_path}, the GCS SPP is NP-hard in general, $\cf$ \cite[Theorem 3.1]{GCS}.

\subsubsection{Convex Relaxation}
A standard convex relaxation replaces \eqref{eq_gcs_shortest_path_binary} with the box constraint $0 \le z_e \le 1$ and converts the implications in \eqref{eq_gcs_shortest_path_edgefeas} into a convex perspective (or conic) formulation; see \cite{GCS_thesis,GCS} for canonical constructions. This relaxation can be tight in many motion-planning instances, but in general it may be loose \cite[Proposition 8.1]{GCS_thesis}. Existing formal tightness guarantees are limited to restrictive settings \cite[Proposition 8.2]{GCS_thesis}. In this paper, we empirically observe tightness in our setting ($\cf$ Section \ref{sec_simulation_results}), but we do not claim a general tightness guarantee.

\subsection{Generalized Potential Games} \label{sec_potential_games}

We consider a game with $N$ vehicles, indexed by $\mcl{N} \coloneqq \{1, \ldots, N\}$. Vehicle $i \in \mcl{N}$ chooses a strategy $\theta_{i}$, which in our driving setting represents a finite-horizon motion plan (trajectory and lane-change actions) and will later be encoded as a path through a vehicle-specific graph of convex sets. Vehicle $i$ incurs a cost $J_{i} (\theta_{i}, \theta_{-i})$, where $\theta_{-i} \coloneqq (\theta_1, \dots, \theta_{i-1}, \theta_{i+1}, \dots, \theta_N)$ denotes the strategies of all other vehicles. In our setting, feasibility depends on other vehicles through coupled constraints, so we use a \textit{generalized} strategy set of the form
\begin{align}\label{generalized_strategy_set}
\Theta_{i} (\theta_{-i}) \coloneqq \{\theta_i \mid (\theta_i,\theta_{-i}) \in \Theta\},
\end{align}
where $\Theta$ is the joint feasible set. Each vehicle's cost is a function $J_i : \Theta \to \mbb{R}$, evaluated on feasible joint strategies $\theta=(\theta_i, \theta_{-i}) \in \Theta$.

Collecting the vehicles, costs, and coupled feasibility constraints yields a game in which each vehicle solves an optimization problem whose feasible set depends on the other vehicles' strategies.

\begin{definition}[Game] \label{def_game}
The  game is denoted by $G$ and is specified by the following set of coupled optimization problems:
\begin{equation} \label{eq_game_def}
G \coloneqq \left\{
\begin{aligned}
\minimize_{\theta_i}\quad & J_i(\theta_i,\theta_{-i})\\
\subjectto \quad & \theta_i \in \Theta_i(\theta_{-i}),
\end{aligned}
\right. \qquad \forall i \in \mcl{N}.
\end{equation}
\end{definition}

A central solution concept for games of this form is the Generalized Nash Equilibrium (GNE), in which no vehicle can unilaterally reduce its cost while satisfying the feasibility constraints imposed by the other vehicles' strategies.

\begin{definition}[Generalized Nash Equilibrium, GNE] \label{def_GNE}
A strategy profile $\theta^\star \in \Theta$ is a GNE of $G$ if, for every vehicle $i \in \mcl{N}$,
\begin{equation} \label{eq_GNE_def}
J_i(\theta_i^\star, \theta_{-i}^\star) \le J_i(\theta_i,\theta_{-i}^\star),
\qquad \forall \theta_i \in \Theta_i(\theta_{-i}^\star).
\end{equation}
\end{definition}

To analyze the existence of equilibria and the behavior of best-response dynamics, it is often useful to identify games whose incentives can be summarized by a single scalar function. This motivates the notion of a generalized potential game, which extends exact potential games to settings with coupled feasibility constraints \cite{pot}.

\begin{definition}[Generalized potential game]\label{def_GPG}
The game $G$ is a generalized potential game if there exists a function $\Phi:\Theta \to \mbb{R}$ such that for all $i \in \mcl{N}$ and all $\theta_{-i}$ for which $\Theta_i(\theta_{-i})\neq\emptyset$, it holds for all $\theta_i,\theta_i'\in \Theta_i(\theta_{-i})$ that
\begin{equation}\label{eq_GPG_def}
J_i(\theta_i',\theta_{-i}) - J_i(\theta_i,\theta_{-i}) = \Phi(\theta_i',\theta_{-i}) - \Phi(\theta_i,\theta_{-i}).
\end{equation}
The function $\Phi$ is called a \emph{generalized potential}.
\end{definition}

We next introduce the iterative best-response (IBR) dynamics that will be used in Section \ref{sec_approach}. Given $\theta_{-i}$, the set of best responses for each vehicle $i \in \mcl{N}$, denoted $\BR_i(\theta_{-i})$, is given by
\begin{equation}\label{eq_best_response}
\BR_i(\theta_{-i}) \coloneqq \argmin_{\theta_i \in \Theta_i(\theta_{-i})} J_i(\theta_i,\theta_{-i}).
\end{equation}
Starting from an initial feasible profile $\theta^{0} \in \Theta$, IBR updates vehicles sequentially. We define
\begin{equation}\label{eq_theta_intermediate}
\theta^{k,i} \coloneqq (\theta_1^{k+1}, \ldots, \theta_i^{k+1}, \theta_{i+1}^{k}, \ldots, \theta_N^{k}),
\qquad i = 1,\ldots,N,
\end{equation}
for the intermediate profile after vehicles $1, \ldots, i$ have been updated during iteration $k$. With this, the update rule is
\begin{equation} \label{eq_IBR_update}
\theta_i^{k+1} \in \BR_i(\theta_{-i}^{k,i-1}), \qquad i=1, \ldots, N, \quad k \ge 0,
\end{equation}
with $\theta^{k,0}=\theta^{k}$ at the beginning of iteration $k$ and $\theta^{k,N}=\theta^{k+1}$ at the end of iteration $k$. In a generalized potential game, any exact best-response update is a descent step for the potential function $\Phi$. A standard result is given below, $\cf$ \cite{pot}.

\begin{proposition}[Potential descent under exact IBR]\label{prop_potential_descent}
Suppose $G$ is a generalized potential game with generalized potential $\Phi$ (Definition \ref{def_GPG}). If, at some intermediate profile $\theta^{k,i-1}$, vehicle $i \in \mcl{N}$ performs an exact best-response update \eqref{eq_best_response}, then
\begin{align*}
\Phi(\theta^{k,i}) \le \Phi(\theta^{k,i-1}).
\end{align*}
Consequently, if all vehicles update with exact best responses, the sequence $\{\Phi(\theta^k)\}_{k\ge 0}$ is monotonically non-increasing along IBR, and therefore convergent, if $\Phi$ is bounded below on $\Theta$.
\end{proposition}

\section{Problem Formulation}\label{sec_problem_formulation}

We consider a multi-vehicle highway driving scenario with $N$ vehicles indexed by $\mcl{N}\coloneqq\{1,\ldots,N\}$ operating on a set of straight lanes $\mcl{L}\subset \mbb{N}$, as illustrated in Fig.~\ref{fig_vehicles}. Planning is performed over a finite-horizon $\mcl{T}\coloneqq\{0,\ldots,T-1\}$ with a discretization step $\Delta t \in \mbb{R}_{> 0}$ and transition index set $\mcl{T}^{-}\coloneqq\{0,\ldots,T-2\}$. Vehicle $i\in\mcl{N}$ has longitudinal position $s_i(t)$, speed $v_i(t)$, acceleration $a_i(t)$, lane index $z_i(t)\in\mcl{L}$, and blinker (lane-change action) $b_i(t)\in\mcl{B}\coloneqq\{-1,0,1\}$, where $b_i(t)=-1$ indicates a lane change to the right, $b_i(t)=0$ indicates staying in lane, and $b_i(t)=1$ indicates a lane change to the left.

\begin{figure}[H]
\hfill
\begin{minipage}[b]{0.48\linewidth}
\centering
\includegraphics[width=\linewidth]{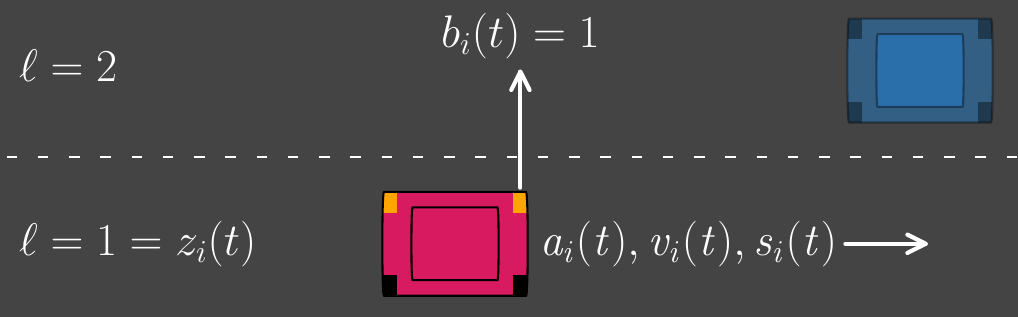}
\end{minipage}
\hfill
\begin{minipage}[b]{0.48\linewidth}
\centering
\includegraphics[width=\linewidth]{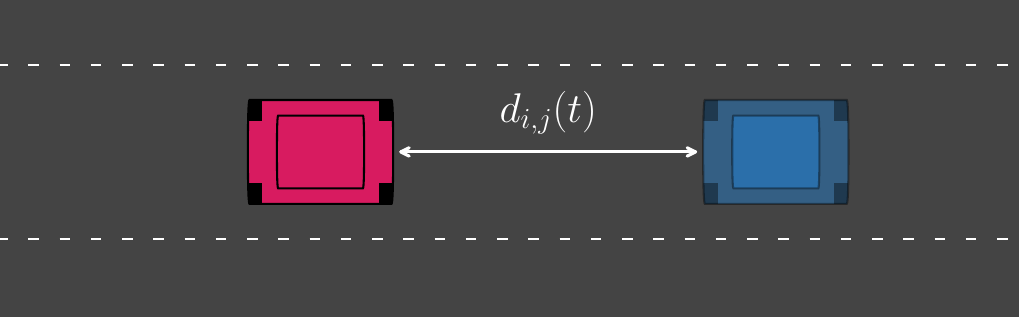}
\end{minipage}
\caption{Vehicles operating on a highway with vehicle-specific variables (left) and longitudinal distance (right).}
\label{fig_vehicles}
\end{figure}

\subsubsection{Dynamics and bounds:}
The motion of vehicle $i\in\mcl{N}$ is modeled by the discrete-time dynamics:
\begin{equation}\label{eq_dynamics}
\begin{aligned}
s_i(t{+}1) &= s_i(t) + \Delta t\, v_i(t),\\
v_i(t{+}1) &= v_i(t) + \Delta t\, a_i(t),
\end{aligned}
\qquad \forall t \in \mcl{T}^{-}.
\end{equation}
We enforce state bounds:
\begin{equation}\label{eq_state_bounds}
s_i(t)\in[\,\ubar{s},\ \obar{s}\,], \qquad
v_i(t)\in[\,\ubar{v}_i,\ \obar{v}_i\,], \qquad \forall t\in\mcl{T},
\end{equation}
where $\ubar{s}$ and $\obar{s}$ denote the endpoints of the road segment and $\ubar{v}_i,\obar{v}_i$ are the minimum and maximum allowable speeds for vehicle $i$. Acceleration is bounded by actuation limits
\begin{equation}\label{eq_actuation_bounds}
a_i(t)\in[\,\ubar{a}_i,\ \obar{a}_i\,], \qquad \forall t\in\mcl{T}^{-},
\end{equation}
with $\ubar{a}_i,\obar{a}_i$ denoting the minimum and maximum admissible accelerations of vehicle $i$.

\subsubsection{Lane evolution:}
Lane changes are restricted to adjacent lanes and are commanded by the blinker:
\begin{equation}\label{eq_lane_evolution}
z_i(t{+}1)= z_i(t)+b_i(t),\qquad
b_i(t)\in\mcl{B},\qquad
z_i(t{+}1)\in\mcl{L},\qquad \forall t\in\mcl{T}^{-}.
\end{equation}

\subsubsection{Relative coordinates:}
For any pair of vehicles $i,j\in\mcl{N}$, define the longitudinal and lateral offsets:
\begin{equation}\label{eq_relative_coords}
d_{i,j}(t)\coloneqq s_j(t)-s_i(t), \qquad
z_{i,j}(t)\coloneqq z_j(t)-z_i(t),
\end{equation}
as shown in Fig.~\ref{fig_vehicles}. The longitudinal offset, in view of \eqref{eq_dynamics}, evolves as:
\begin{equation}\label{eq_relative_dynamics}
d_{i,j}(t{+}1)=d_{i,j}(t)+\Delta t\bigl(v_j(t)-v_i(t)\bigr), \qquad \forall t\in\mcl{T}^{-}.
\end{equation}

\subsubsection{Safety rules:}
To prevent collisions, we impose longitudinal safety for vehicles in the same lane and lateral safety for
vehicles in adjacent lanes.

\paragraph{Rule 1 (Longitudinal safety in the same lane).}
If $z_{i,j}(t)=0$, then vehicles must maintain a minimum longitudinal separation, $\ie$,
\begin{equation}\label{eq_rule1}
|d_{i,j}(t)| \ge d_i^{\saf}, \qquad \forall t\in\mcl{T},
\end{equation}
where $d_i^{\saf}>0$ is the required safety gap for vehicle $i$.

\paragraph{Rule 2 (Lateral safety across adjacent lanes).}
If vehicles are side-by-side in adjacent lanes, $\ie$,
\begin{equation}\label{eq_rule2_condition}
|d_{i,j}(t)| \le d_i^{\saf}, \qquad |z_{i,j}(t)|=1, \qquad \forall t\in\mcl{T},
\end{equation}
then simultaneous swaps into each other's lanes are forbidden:
\begin{equation}\label{eq_rule2}
z_i(t{+}1)\neq z_j(t), \qquad z_j(t{+}1)\neq z_i(t), \qquad \forall t\in\mcl{T}^{-}.
\end{equation}

\subsubsection{Preferences and objectives:}
Each vehicle $i\in\mcl{N}$ is assigned a desired cruising speed $v_i^{\des}\in[\,\ubar{v}_i,\ \obar{v}_i\,]$ and a desired lane $\ell_i^{\des}\in\mcl{L}$. We also fix weights $w_{i,v},w_{i,\ell},w_{i,a},w_{i,b}\in\mbb{R}_{> 0}$ that trade off speed tracking,
lane preference, control effort, and lane-change usage, respectively. These parameters are used to define
the edge and terminal costs in the GCS construction in Section~\ref{sec_approach}.

\subsection{Individual vehicle strategy and optimization problem}\label{subsec:single_vehicle}

A vehicle $i$'s finite-horizon strategy is the collection of its state and input trajectories,
\begin{equation}\label{eq:theta_i_def}
\theta_i \coloneqq \Big( \{ s_i(t), v_i(t), z_i(t) \}_{t\in\mcl{T}},\ \{ a_i(t), b_i(t) \}_{t \in \mcl{T}^{-}} \Big),
\end{equation}
and we write $\theta \coloneqq( \theta_1, \ldots, \theta_N)$ for the joint strategy profile. Given initial conditions $\{s_i(0), v_i(0), z_i(0) \}_{i \in \mcl{N}}$, each vehicle $i \in \mcl{N}$ seeks a feasible strategy that respects the shared safety rules while optimizing its own preferences ($\eg$, tracking a desired speed $v_i^{\des}$ and lane $\ell_i^{\des}$, as introduced earlier). We formalize through the following assumption.

\begin{assumption}[Separable objectives]\label{assum_separable_cost}
Vehicle objectives are separable across vehicles, $\ie$, $J_i(\theta_i) = J_i(\theta_i,\theta_{-i})$, and interaction between vehicles occurs only through the coupled feasibility constraints (Rules~1--2).
\end{assumption}

Assumption~\ref{assum_separable_cost} captures the modeling choice that vehicles do not directly penalize (or reward) other vehicles in their objectives; rather, strategic coupling arises solely because each vehicle's feasible set depends on the others through shared safety rules. Formally, vehicle $i$ solves:
\begin{equation} \label{eq_single_vehicle}
\begin{aligned}
&\minimize_{\theta_i}\quad && J_i(\theta_i)\\
&\subjectto \quad
&& s_i(0),\,v_i(0),\,z_i(0)\ \text{given},\\
&&& \text{\cref{eq_dynamics,eq_state_bounds,eq_actuation_bounds,eq_lane_evolution}},\\
&&& \eqref{eq_rule1},\ \eqref{eq_rule2}, \qquad & \forall j \in\mcl{N}\setminus\{i\}.
\end{aligned}
\end{equation}
The constraints in Equations \cref{eq_dynamics,eq_state_bounds,eq_actuation_bounds,eq_lane_evolution} enforce discrete-time longitudinal dynamics \eqref{eq_dynamics}, bounded roadway and speed domains \eqref{eq_state_bounds}, actuation limits \eqref{eq_actuation_bounds}, and lane-change feasibility \eqref{eq_lane_evolution}. Vehicle interactions enter only through the coupled collision-avoidance Rules~1 and 2: \eqref{eq_rule1} enforces longitudinal separation between vehicles in the same lane, while \eqref{eq_rule2} prevents simultaneous lane swaps when vehicles are side-by-side in adjacent lanes.

To connect each vehicle's optimization problem in \eqref{eq_single_vehicle} with the GCS construction in Section~\ref{sec_approach}, we impose the following mild modeling assumptions.

\begin{assumption}[Convexity of continuous costs] \label{assum_convex_cost}
For each vehicle $i\in\mcl{N}$ and for any fixed lane and blinker sequences $(z_i,b_i)$, the objective $J_i(\theta_i)$ is proper, closed, and convex in the continuous variables $\{s_i(t), v_i(t)\}_{t \in \mcl{T}}$ and $\{a_i(t)\}_{t \in \mcl{T}^{-}}$, and takes values in $\mbb{R}_{\ge 0}$.
\end{assumption}

Assumption~\ref{assum_convex_cost} ensures that once the discrete lane/blinker sequence is fixed, the remaining optimization over continuous variables is a convex program. This property is essential for modeling each maneuver option as a convex set (vertex) and each time-adjacent transition as a convex constraint (edge) in a GCS.

\paragraph{Example objective function.}
A concrete cost function consistent with the preferences and assumptions introduced earlier, and used later in Section~\ref{sec_approach}, is the following cost function:
\begin{equation}\label{eq_Ji_example}
\begin{aligned}
J_i(\theta_i)
\coloneqq &\sum_{t \in \mcl{T}^{-}} \Bigl[
w_{i,v}\bigl(v_i(t{+}1)-v_i^{\des}\bigr)^2
+ w_{i,\ell}\bigl(z_i(t{+}1)-\ell_i^{\des}\bigr)^2 \\
&\hspace{3.2em}
+\, w_{i,a}\,a_i(t)^2
+ w_{i,b}\,b_i(t)^2
\Bigr]
+ w_{i,v}\bigl(v_i(T{-}1)-v_i^{\des}\bigr)^2 ,
\end{aligned}
\end{equation}
which penalizes deviations from the desired speed and lane, as well as acceleration and lane changes. For any fixed lane and blinker sequences $(z_i,b_i)$, the cost in \eqref{eq_Ji_example} is convex in the continuous variables $(s_i,v_i,a_i)$.

Let $\Theta$ denote the joint feasible set of all strategy profiles $\theta$ satisfying the dynamics, bounds, lane evolution, and safety rules for every vehicle. Under Assumption~\ref{assum_separable_cost}, the induced game admits a generalized potential formulation (Definition~\ref{def_GPG}) with potential:
\begin{equation}\label{eq_potential}
\Phi(\theta)\coloneqq \sum_{i\in\mcl{N}} J_i(\theta_i).
\end{equation}

\begin{problem}\label{prob_multi_vehicle}
Given initial conditions $\{s_i(0),v_i(0),z_i(0)\}_{i\in\mcl{N}}$ and vehicle preferences ($\eg$, $v_i^{\des}$ and $\ell_i^{\des}$), find a feasible joint strategy profile $\theta^\star\in\Theta$ such that no vehicle can unilaterally decrease its cost while respecting the coupled feasibility constraints, $\ie$, $\theta^\star$ is a generalized Nash equilibrium of the induced game.
\end{problem}

\section{Iterative Best Response Graphs of Convex Sets} \label{sec_approach}

We address Problem \ref{prob_multi_vehicle} using IBR. At each iteration, vehicles update their strategies sequentially by solving a single-vehicle motion planning problem while treating the trajectories of all other vehicles as fixed.

For vehicle $i \in \mcl{N}$, we construct a directed graph of convex sets $\mcl{G}_{i} = (\mcl{V}_i, \mcl{E}_i)$ tailored to the current best-response subproblem. Vertices represent collision-free convex regions at each time step, and directed edges represent dynamically feasible transitions between time-adjacent vertices.

\begin{figure}[H]
\begin{minipage}[t]{0.70\linewidth}
\includegraphics[
width=\linewidth
]{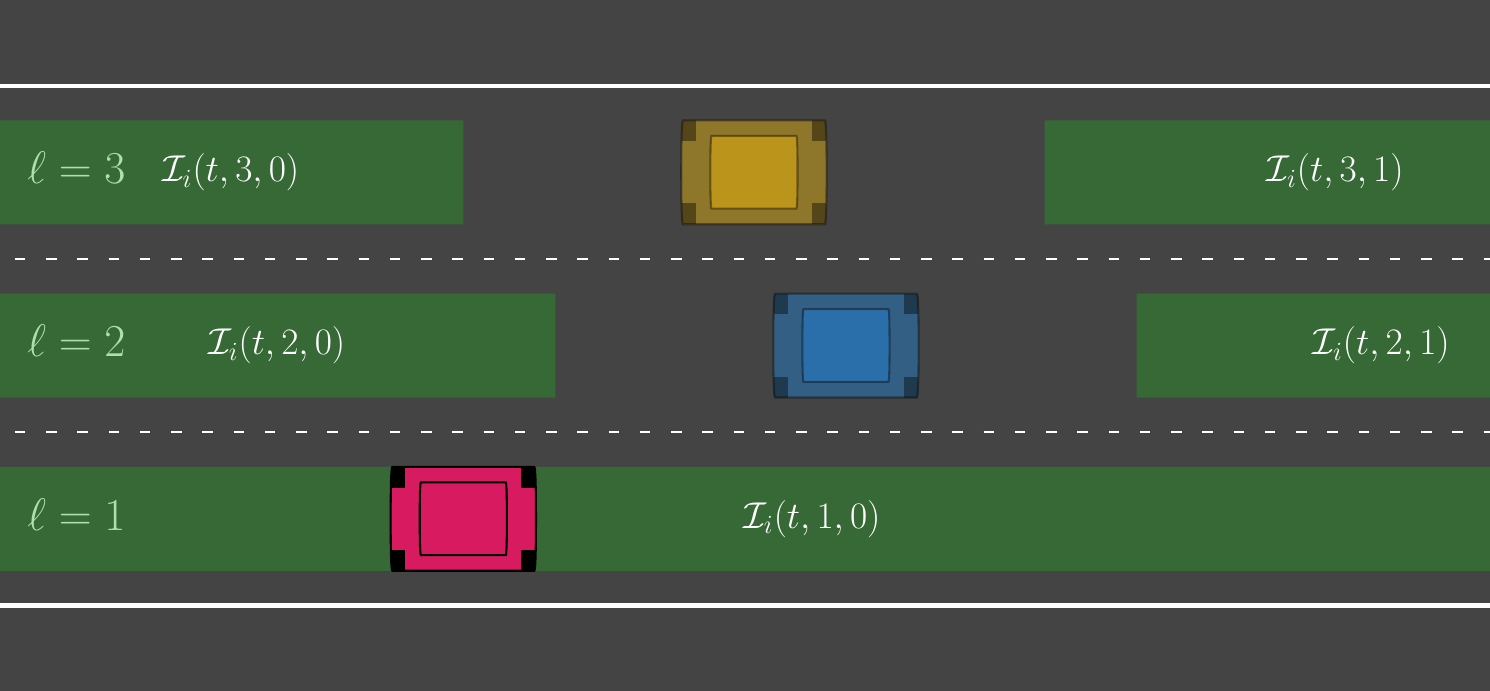}
\caption{A set of vehicles driving on a highway with the red vehicle's safe gaps in green.}
\label{fig_safe_gaps}
\end{minipage}
\hspace{0.03\linewidth}
\begin{minipage}[t]{0.20\linewidth}
\includegraphics[
width=\linewidth
]{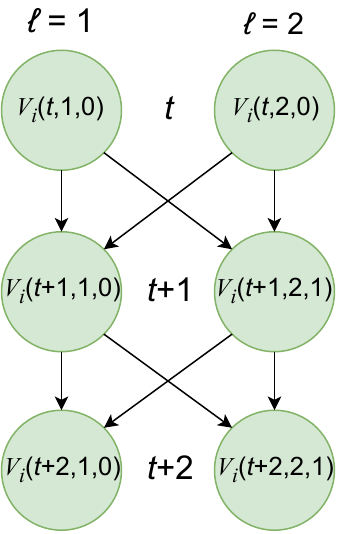}
\caption{Time expanded graph.}
\label{fig_safe_gaps_graph}
\end{minipage}
\end{figure}

\subsection{Vertex Construction} \label{sec_vertex}

Fix an IBR iteration and treat the other vehicles' trajectories as known. For each lane $\ell \in \mcl{L}$ and time step $t \in \mcl{T}$, define the unsafe longitudinal interval around each vehicle $j \neq i$ that is in lane $\ell$ at time step $t \in \mcl{T}$:
\begin{align}
\mcl{K}_{i}^{j}(\ell,t) \coloneqq \bigl(\, s_j(t) - d_i^{\saf},\ s_j(t) + d_i^{\saf} \,\bigr).
\end{align}
Let the union of unsafe intervals in lane $\ell$ at time step $t \in \mcl{T}$ be
\begin{align}
\mcl{F}_{i}(\ell,t) \coloneqq \bigcup_{\substack{j\in\mcl{N}\setminus\{i\}\\ z_j(t)=\ell}} \mcl{K}_{i}^{j}(\ell,t).
\end{align}
Over a bounded road segment $[\,\ubar{s},\ \obar{s}\,]$, the collision-free set is
\begin{align}
\mcl{S}_{i}(\ell,t) \coloneqq [\,\ubar{s},\ \obar{s}\,] \setminus \mcl{F}_{i}(\ell,t).
\end{align}
Since $\mcl{F}_{i}(\ell,t)$ is a finite union of open intervals, $\mcl{S}_{i}(\ell, t)$ can be expressed as a finite union of disjoint closed intervals (safe gaps)
\begin{align}
\mcl{S}_{i}(\ell,t) = \bigcup_{g \in \mcl{H}_{i}(t, \ell)} \mcl{I}_{i}(t,\ell,g),
\end{align}
where $\mcl{H}_{i}(t,\ell)$ is the finite index set of connected components of $\mcl{S}_{i}(\ell, t)$ and $\{ \mcl{I}_{i}(t,\ell,g) \}_{g \in \mcl{H}_{i} (t, \ell) }$ denotes the corresponding family of pairwise-disjoint closed intervals. Equivalently, each $g \in \mcl{H}_{i} (t, \ell)$ labels one collision-free longitudinal ``gap'' in lane $\ell \in \mcl{L}$ at time step $t \in \mcl{T}$ between consecutive unsafe regions induced by other vehicles.

Each safe interval $\mcl{I}_{i}(t,\ell,g)$, as illustrated in Fig.~\ref{fig_safe_gaps}, induces a vertex $V_i(t,\ell,g)\in\mcl{V}_i$  whose continuous state variable is
\begin{align}
x_{V_i(t,\ell,g)} \coloneqq
\begin{bmatrix}
s_{V_i(t,\ell,g)}\\
v_{V_i(t,\ell,g)}
\end{bmatrix} \in\mbb{R}^2,
\end{align}
and whose convex constraints are
\begin{equation}\label{eq_vertex_constraints}
s_{V_i(t,\ell,g)} \in \mcl{I}_{i}(t,\ell,g),
\qquad
v_{V_i(t,\ell,g)} \in [\,\ubar{v}_i,\ \obar{v}_i\,].
\end{equation}
The discrete lane choice is encoded by the vertex lane index $\ell \in \mcl{L}$.

\paragraph{Initial condition.}
To impose initial state conditions of vehicle $i$, choose any $g_0 \in \mcl{H}_{i}\big(0,z_i(0)\big)$ such that $s_i(0)\in \mcl{I}_{i}(0,z_i(0),g_0)$ (ties may be broken arbitrarily), and enforce
\begin{equation}\label{eq_initial_condition_vertex}
x_{V_i(0,z_i(0),g_0)}=
\begin{bmatrix}
s_i(0)\\
v_i(0)
\end{bmatrix}.
\end{equation}
This fixes a valid source vertex consistent with the initial state.

\subsection{Edge Construction}\label{sec_edge}

Directed edges connect time-adjacent vertices:
\begin{align}
e = \bigl(V_i(t,\ell,g),\ V_i(t+1,\ell',g')\bigr),
\qquad t\in\mcl{T}^{-},
\end{align}
as shown in Fig.~\ref{fig_safe_gaps_graph}. Only edges between adjacent lanes are permitted, $\ie$, $|\ell'-\ell|\le 1$. Each edge enforces dynamic feasibility of the form \eqref{eq_dynamics} with a decision variable $a_i(t) \in [\,\ubar{a}_i,\ \obar{a}_i\,]$ satisfying \eqref{eq_actuation_bounds} and consistent state variables at times $t$ and $t+1$. The lane-change action on the edge is $b_i(t)=\ell'-\ell\in\{-1,0,1\}$, consistent with \eqref{eq_lane_evolution}. Lateral safety (Rule~2) is enforced by excluding lane-change edges that would violate \eqref{eq_rule2} given the fixed trajectories and lane-change actions of other vehicles. Each edge $e=\bigl(V_i(t,\ell,g),\,V_i(t+1,\ell',g')\bigr)$ is assigned a convex quadratic edge cost $c_e(\cdot)$ that depends on the decision variables on that edge, namely the successor velocity $v_{V_i(t+1,\ell',g')}$ and the control $a_i(t)$ (with $b_i(t)=\ell'-\ell$):
\vspace{-10pt}
\begin{multline}\label{eq_edge_cost}
c_e \bigl(x_{V_i(t,\ell,g)},x_{V_i(t+1,\ell',g')},a_i(t)\bigr) =
w_{i,v}\!\left(v_{V_i(t+1,\ell',g')}-v_i^{\des}\right)^{2}
+ w_{i,\ell}\!\left(\ell'-\ell_i^{\des}\right)^{2} \\
+
w_{i,a}\,a_i(t)^{2}
+ w_{i,b}\,b_i(t)^{2},
\end{multline}
where $w_{i,v}, w_{i,\ell}, w_{i,a}, w_{i,b}\in\mbb{R}_{>0}$ are weights. A terminal penalty may be included as a vertex cost at $t=T-1$, $\eg$,
\begin{equation}
c_V \left( x_{V_i(T-1,\ell,g)} \right) = w_{i,v} \left( v_{V_i(T-1,\ell,g)}-v_i^{\des} \right)^2.
\end{equation}

Equipped with this construction, vehicle $i \in \mcl{N}$'s single-vehicle planning problem \eqref{eq_single_vehicle} can be cast as a GCS shortest path problem on $\mcl{G}_i$.

\begin{proposition}[Single-vehicle optimality under tightness]\label{prop_single_vehicle_tight}
Fix the trajectories of all vehicles $j\neq i$ and construct $\mcl{G}_i$ accordingly. If the convex relaxation of the induced GCS shortest path problem is tight, $\ie$, it admits an optimal solution with $z_e \in \{0,1\}$, then solving the relaxed problem yields a globally optimal solution of vehicle $i$'s mixed-integer best-response subproblem on $\mcl{G}_i$. Consequently, under tightness, the resulting update is an \emph{exact} best response, so each IBR step is a potential-descent step in the sense of Proposition~\ref{prop_potential_descent}.
\end{proposition}

\begin{remark}
If the relaxation is not tight, the computed update can be interpreted as an approximate best response. Section \ref{sec_error} quantifies the effect of such inexactness.
\end{remark}

\subsection{Iterative Best-Response on Graphs of Convex Sets Algorithm} \label{sec_IBR}

We solve the multi-vehicle game in Problem \ref{prob_multi_vehicle} using an IBR scheme, summarized in Algorithm \ref{alg_ibr}.
\begin{algorithm}
\caption{Iterative Best-Response on Graphs of Convex Sets (\ours)} \label{alg_ibr}
\SetKwInput{KwInput}{Input}
\SetKwInput{KwOutput}{Output}

\KwInput{Initial feasible strategies $\theta^{0}=(\theta_1^{0},\ldots,\theta_N^{0})\in\Theta$, maximum sweeps $K_{\max}$, tolerance $\epsilon$}
\KwOutput{Final strategy profile $\theta^{\mathrm{out}}=(\theta_1^{\mathrm{out}},\ldots,\theta_N^{\mathrm{out}})$}

Set $k\gets 0$ and compute $\Phi(\theta^{0})$\;
\For{$k = 0$ \KwTo $K_{\max}-1$}{
Set $\theta^{k,0}\gets \theta^{k}$ \tcp*{start of sweep}
\For{$i = 1$ \KwTo $N$}{
Construct vehicle $i$'s GCS $\mathcal{G}_i^k$ using the other vehicles' strategies $\theta_{-i}^{k,i-1}$\; \label{alg_construct}
Compute a (possibly inexact) best response by solving the single-vehicle GCS problem on $\mathcal{G}_i^k$ to obtain $\theta_i^{k+1}$\; \label{alg_solve}
Update the intermediate profile $\theta^{k,i}\gets(\theta_1^{k+1},\ldots,\theta_i^{k+1},\theta_{i+1}^{k},\ldots,\theta_N^{k})$\; \label{alg_update}
}
Set $\theta^{k+1}\gets \theta^{k,N}$ \tcp*{end of sweep} \label{alg_end_of_sweep}
Compute potential change $\Delta \gets \bigl|\Phi(\theta^{k+1})-\Phi(\theta^{k})\bigr|$\; \label{alg_pot}
\If{$\Delta < \epsilon$}{
\textbf{break}\; \label{alg_term}
}
}
\Return{$\theta^{\mathrm{out}} \gets \theta^{k+1}$}\;
\end{algorithm}

The algorithm proceeds in \emph{sweeps} indexed by $k$: during sweep $k$, vehicles update sequentially, and each vehicle $i$ computes a best response to the most recent strategies of the other vehicles, denoted $\theta_{-i}^{k,i-1}$. Concretely, at its update, vehicle $i$ constructs a GCS instance $\mathcal{G}_i^k$ that encodes collision-free regions and feasible transitions given $\theta_{-i}^{k,i-1}$, solves the resulting single-vehicle GCS SPP, and records the resulting strategy $\theta_i^{k+1}$ (lines \ref{alg_construct}--\ref{alg_update} of Algorithm \ref{alg_ibr}). After all vehicles have updated, the sweep terminates, and the algorithm produces the next joint strategy profile $\theta^{k+1}$ (line~\ref{alg_end_of_sweep}).

Under exact best-response updates in a generalized potential game (see Section~\ref{sec_potential_games}), Proposition~\ref{prop_potential_descent} implies that the generalized potential is monotonically nonincreasing along the intermediate profiles $\theta^{k,i}$, and hence along the sweep iterates $\theta^{k}$. In our implementation, each best response is computed by solving a convex relaxation of the underlying GCS mixed-integer problem. Consequently, the update may be inexact and should be interpreted as an approximate best response. The resulting equilibrium error and its relationship to the per-update suboptimality are quantified in Section \ref{sec_error}. We terminate IBR when the decrease in the potential across a sweep is below a specified tolerance (lines \ref{alg_end_of_sweep}--\ref{alg_term}).

\subsection{Error Quantification} \label{sec_error}

This section quantifies the extent to which inexact best-response updates degrade the limiting quality of the strategy profiles produced by the IBR procedure. Since each vehicle solves a relaxed GCS problem in practice, the resulting update may not be an exact best response to the other vehicles' fixed strategies. We therefore measure the suboptimality of each update relative to the true best-response problem and translate this into an approximate equilibrium guarantee.

We recall the definitions of the joint feasible set $\Theta$, the individual feasible sets $\Theta_i(\theta_{-i})$, and the intermediate profiles $\theta^{k,i}$ introduced in Equations \eqref{generalized_strategy_set} and \eqref{eq_theta_intermediate}. We note that for the remainder of this section, we revert to the cost function notation used prior to Assumption \ref{assum_separable_cost}, $\ie$, $J_i(\theta_i,\theta_{-i}) = J_i(\theta_i)$.

\paragraph{Approximate best responses.}
At each update, vehicle $i \in \mcl{N}$ aims to minimize its cost over the feasible set induced by the current strategies of the other vehicles. We measure the quality of the computed update by comparing its cost to the optimal value of this best-response problem.

\begin{definition}[Approximate best response]\label{def_ABR}
At sweep $k$ and vehicle $i \in \mcl{N}$, we say the computed update $\theta_i^{k+1}$ has best-response error at most $\epsilon_i^k\ge 0$ if
\begin{equation}\label{ineq_ABR}
J_i(\theta_i^{k+1}, \theta_{-i}^{k,i-1}) \le \inf_{\theta_i \in \Theta_i(\theta_{-i}^{k,i-1})} J_i(\theta_i, \theta_{-i}^{k,i-1}) +\epsilon_i^k.
\end{equation}
\end{definition}

\paragraph{Approximate equilibrium.}
A natural benchmark for IBR is a generalized Nash equilibrium (Definition~\ref{def_GNE}), in which no vehicle can reduce its cost through a feasible unilateral deviation. Since we allow inexact best responses, we use the standard relaxation of this notion, $\cf$ \cite[Definition 3.4.7]{shoham2008multiagent}.

\begin{definition}[$\epsilon$-GNE] \label{def_epsGNE}
A strategy profile $\theta\in\Theta$ is an $\epsilon$-GNE if
\begin{align}
J_i(\theta_i,\theta_{-i}) \le \inf_{\hat{\theta}_i\in\Theta_i(\theta_{-i})} J_i(\hat{\theta}_i,\theta_{-i}) + \epsilon, \qquad \forall i \in \mcl{N}.
\end{align}
\end{definition}

To connect update-level errors to equilibrium quality, we introduce each vehicle's \emph{regret}, $\ie$, the gap between its current cost and its best feasible unilateral deviation.

\begin{align}
r_i(\theta) \coloneqq J_i(\theta_i, \theta_{-i}) - \inf_{\hat{\theta}_i \in \Theta_i(\theta_{-i})} J_i(\hat{\theta}_i, \theta_{-i}) \ge 0.
\end{align}

Our analysis is based on two standard technical assumptions. The first prevents the potential from decreasing without bound, while the second postulates that the inexactness of each best-response computation is uniformly bounded.

\begin{assumption}\label{assum_bounded}
The potential function $\Phi$ is bounded below on $\Theta$, and $\Theta$ is nonempty and compact.
\end{assumption}

\begin{assumption}\label{assum_bderror}
The IBR iterates satisfy \eqref{ineq_ABR} with errors uniformly bounded by $\epsilon_i^k \le \bar{\epsilon} < \infty$ for all $i \in \mcl{N}$ and all $k\ge 0$.
\end{assumption}

Assumption \ref{assum_bounded} is standard in our setting, as physical constraints and finite discrete decisions ensure the feasible strategy space is a nonempty compact set over which the nonnegative potential function is bounded below. However, Assumption \ref{assum_bderror} is stronger and generally difficult to verify, requiring a uniform bound of the suboptimality of the GCS convex relaxation that is often observed in motion-planning problems (see Remark \ref{rem_tightness_empirical}), but is difficult to guarantee theoretically ($\cf$, \cite[Proposition~8.1]{GCS_thesis}).

The following result states that if each IBR update is an $\epsilon$-approximate best response with a uniform error bound, then the iterates asymptotically lie in an approximate equilibrium set: no vehicle can improve by more than $\bar\epsilon$ in the limit. 

\begin{theorem}[$\bar\epsilon$-equilibrium under \ours]\label{thm_approximateGNE}
Suppose $G$ is a generalized potential game with potential $\Phi$ (Definition \ref{def_GPG}). Under Assumptions \ref{assum_bounded}--\ref{assum_bderror},
\begin{align*}
\limsup_{k\to\infty}\ \max_{i \in \mcl{N}} r_i(\theta^k)\ \le\ \bar{\epsilon}.
\end{align*}
In particular, the iterates asymptotically lie in the $\bar\epsilon$-GNE set (Definition \ref{def_epsGNE}).
\end{theorem}

\begin{proof}
Fix a sweep index $k$ and consider the update of vehicle $i\in\mcl{N}$ from the intermediate profile $\theta^{k,i-1}$ to $\theta^{k,i}$. Since $\Phi$ is a generalized potential (Definition~\ref{def_GPG}), changes in vehicle $i$'s cost match changes in $\Phi$ under feasible unilateral deviations. Combining this property with the approximate best-response condition \eqref{ineq_ABR} yields
\begin{align}
\Phi(\theta^{k,i-1})-\Phi(\theta^{k,i})
&= J_i(\theta_i^{k,i-1},\theta_{-i}^{k,i-1}) - J_i(\theta_i^{k,i},\theta_{-i}^{k,i-1}) \nonumber\\
&\ge r_i(\theta^{k,i-1}) - \epsilon_i^k
\ \ge\ r_i(\theta^{k,i-1}) - \bar{\epsilon}.
\label{ineq_improv}
\end{align}

We proceed by contradiction. Suppose there exists $\eta>0$ and an infinite subsequence of sweeps $\{k_m\}_{m \in \mbb{N}}$ such that, for each $m \in \mbb{N}$, there exists an index $i_m \in \mcl{N}$ with
\begin{equation}\label{ineq_contradiction}
r_{i_m}(\theta^{k_m,i_m-1}) \ge \bar{\epsilon} + \eta.
\end{equation}
Applying \eqref{ineq_improv} at the update of vehicle $i_m$ during sweep $k_m$ gives
\begin{align}
\Phi(\theta^{k_m,i_m-1})-\Phi(\theta^{k_m,i_m}) \ge r_{i_m}(\theta^{k_m,i_m-1})-\bar{\epsilon} \ge \eta.
\end{align}
In particular, since $\theta^{k_m,i_m}$ occurs within sweep $k_m$, we have $\Phi(\theta^{k_m+1})\le \Phi(\theta^{k_m})-\eta$. Iterating along the subsequence implies $\Phi(\theta^{k_m}) \to -\infty$ as $m\to\infty$, contradicting the boundedness of $\Phi$ below on $\Theta$ (Assumption~\ref{assum_bounded}). Therefore,
\begin{align}
\limsup_{k\to\infty}\max_{i\in\mcl{N}} r_i(\theta^k)\le \bar{\epsilon},
\end{align}
which proves the claim. \hfill $\square$
\end{proof}

\begin{remark}    
As noted earlier, GCS relaxations are often tight in practice ($\ie$, $\epsilon_i^k=0$), but they can be arbitrarily loose in special cases \cite[Proposition 8.1]{GCS_thesis}, in which case Assumption \ref{assum_bderror} may fail. For instance, this may happen when the costs and edge weights are symmetric \cite{GCS_thesis,GCS}. In such cases, the result above should be interpreted as a conditional guarantee: whenever the per-update best-response error remains uniformly bounded, the IBR iterates converge to a correspondingly bounded approximate equilibrium set.
\end{remark}

\section{Simulation Results} \label{sec_simulation_results}
The performance of the proposed method is illustrated in a representative driving scenario. The parameters for the driving problem are $N = 6$, $|\mcl{L}| = 4$, $T = 30$, and discretization step of $\Delta t = \SI{0.3}{\second}$. The scenario is set up with two vehicles merging from a terminating lane into a three-lane highway, where four other vehicles operate. The algorithm converges in 2 iterations, with the resulting equilibrium visualized in Fig. \ref{fig_scenario_snapshots}. All simulations were performed using the \texttt{GCSOPT} library~\cite{gcsopt} on a laptop with a AMD Ryzen\texttrademark~7~5800U (8-core, 16-thread) CPU and solved using the MOSEK~11.0.29 solver~\cite{mosek}, leading to a total wall-clock time of $\SI{297.24}{\second}$ for this scenario. The resulting strategies $\theta^{\star}$ drive each vehicle $i \in \mcl{N}$ safely over the chosen time horizon $\mcl{T}$ to a target lane $\ell_i^\text{des}$, while tracking a desired velocity $v_i^\text{des}$. In doing so, the vehicles perform multi-layered overtaking maneuvers while accelerating and ensuring both lateral and longitudinal safety.

\begin{figure}[H]
\includegraphics[width=\textwidth]{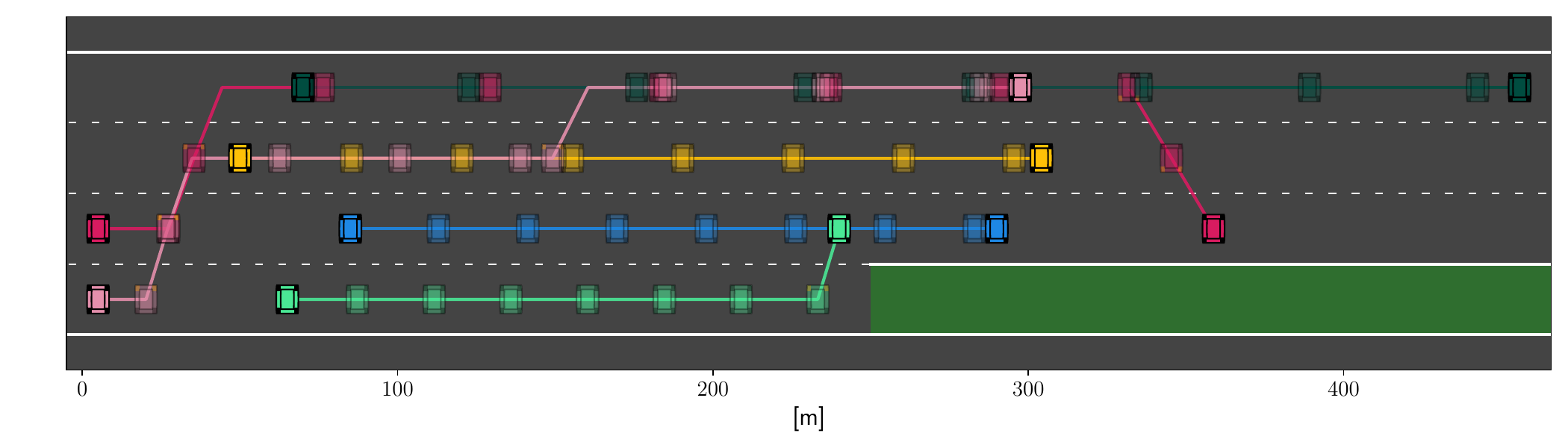}
\caption{Trajectories of six vehicles on a highway with one merging lane. The initial and final positions of each vehicle are shown without transparency, while positions at every third time step are shown with transparency.} \label{fig_scenario_snapshots}
\end{figure}
Additionally, multiple randomized simulations with $N=4$, $|\mcl{L}| = 3$, $T=30$, and $\Delta t = \SI{0.3}{\second}$ were conducted, with parameters sampled from the following uniform distributions: 
\begin{align*}
&v_i^\text{des} \sim \mcl{U}\ (\tfrac{80}{3.6}, \tfrac{160}{3.6}),  
&&\ell_i^\text{des} \sim \mcl{U}(\{1,\dots,|\mcl{L}|\}),
& w_{i,v} \sim \mcl{U}(0.1,1.0), \\
&w_{i,\ell} \sim \mcl{U}(5,25), 
&&w_{i,b} \sim \mcl{U}(5,10),
&w_{i,a} \sim \mcl{U}(0.1,0.5).
\end{align*}
Initial states and lane configurations were randomly sampled according to: 
\begin{align*}
&s_i^{0} \sim \mcl{U}\ ({0.0}, {200}),
&&v_i^{0} \sim \mcl{U}\ (\tfrac{60}{3.6}, \tfrac{130}{3.6}),
&\ell_i^{0} \sim \mcl{U}(\{1,2,3\}),
\end{align*}
until a collision-free arrangement was obtained. Figure \ref{fig_Pot_fctn} illustrates the evolution of the potential function $\Phi(\theta^k)$ over the iterations $k$. As predicted, the potential function values are decreasing, with a strictly positive decrease. Consequently, each iteration of the algorithm moves the system closer to a $\epsilon$-GNE, providing empirical validation of the proposed approach, \ours.

\begin{figure}[H]
\includegraphics[width=\textwidth]{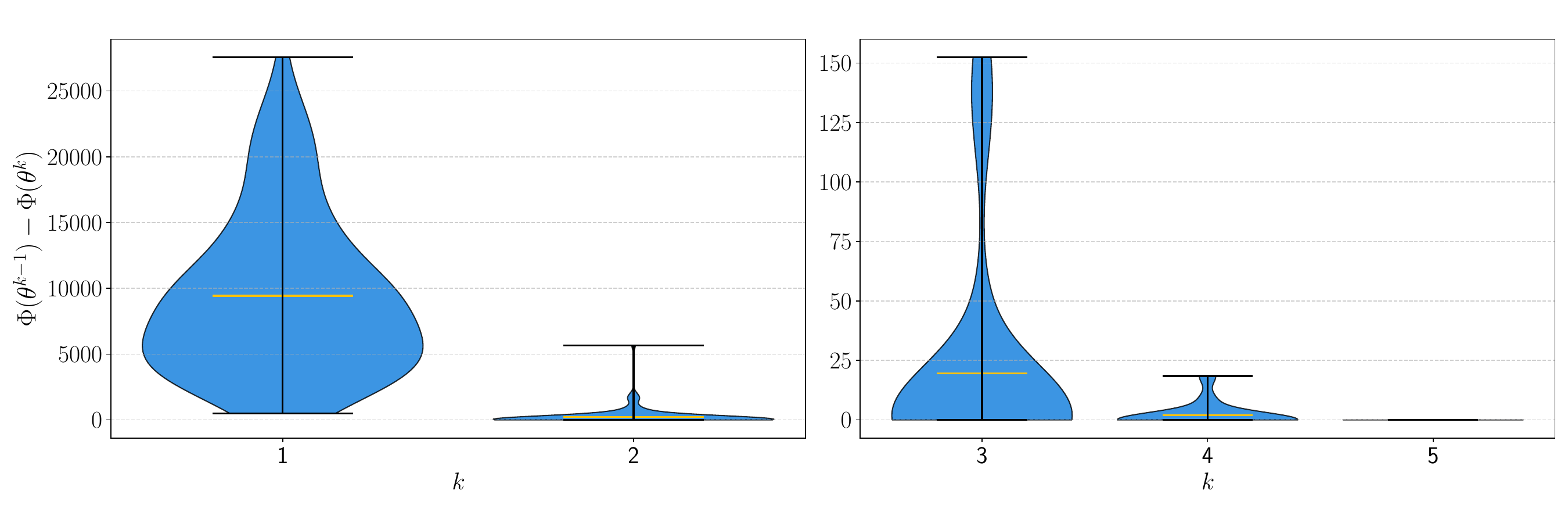}
\caption{Distributions of potential function evolution over iterations for 100 randomized simulation setups with the mean in orange. As shown, the potential function is always smaller at iteration $k$ than at $k-1$.} \label{fig_Pot_fctn}
\end{figure}

\begin{remark} \label{rem_tightness_empirical}
In all our simulations, the convex relaxation was observed to be very tight,  enabling us to solve complex driving scenarios to globally optimal solutions. While our work provides no theoretical tightness guarantee, these empirical observations suggest that additional structure may be present and motivate future work aimed at establishing formal tightness guarantees. \end{remark}

\section{Conclusion}\label{sec_conclusion}

This paper presented \ours, a graphs of convex sets (GCS) formulation for multi-vehicle highway autonomous driving with strategic interaction. For each vehicle, collision-free maneuver options and dynamically feasible transitions are encoded as a shortest-path problem (SPP) on a vehicle-specific GCS. This representation separates continuous trajectory optimization (handled within convex vertex/edge constraints and costs) from discrete maneuver selection (encoded in the graph topology) and admits a convex relaxation that is often tight. Since each vehicle's feasible maneuver graph depends on the current strategies of the other vehicles, a single centralized GCS is not naturally available; instead, \ours\ employs an iterative best-response (IBR) procedure in which vehicles repeatedly solve their individual GCS subproblems while holding the strategies of the other vehicles fixed.

From a game-theoretic perspective, the coupled feasibility constraints induce a generalized Nash game with a generalized potential structure that can be exploited to interpret IBR as an (approximate) descent method on a global potential function. Under exact best responses, the potential is monotonically non-increasing (Proposition~\ref{prop_potential_descent}). When best responses are computed inexactly ($\eg$, due to relaxation looseness), our error analysis provides conditions under which the iterates converge to a neighborhood of the generalized Nash equilibrium set, with a bound on the maximum unilateral improvement available to any vehicle (Theorem~\ref{thm_approximateGNE}).

Several practical considerations affect performance and the particular equilibrium reached. For example, the limiting equilibrium point can depend on initialization and on the vehicle update order. In the current implementation, each vehicle reconstructs its GCS at every update, even though successive graphs often change only locally as other vehicles move. An important direction for future work is to warm-start and reuse graph structure across sweeps ($\eg$, caching safe-gap decompositions and updating only affected vertices/edges), which could substantially reduce per-iteration overhead. Another direction is to incorporate randomized or priority-based update schedules to mitigate order sensitivity, and to refine termination criteria based on per-vehicle regret (Section~\ref{sec_error}) rather than only on potential decrease.

Finally, while convex relaxations are frequently tight in practice, relaxation looseness can degrade best-response accuracy. Developing tighter relaxations and problem-specific certificates of tightness for driving instances remains an important avenue for future work.

\section*{Acknowledgement}
This work was supported in part by the National Science Foundation under Grants 2211548 and 2336840.

\bibliographystyle{splncs04}
\bibliography{ref}

@ARTICLE{game_drive,
author={Hang, Peng and Lv, Chen and Xing, Yang and Huang, Chao and Hu, Zhongxu},
journal={IEEE Transactions on Intelligent Transportation Systems}, 
title={{Human-Like Decision Making for Autonomous Driving: A Noncooperative Game Theoretic Approach}}, 
year={2021},
volume={22},
number={4},
pages={2076-2087},
doi={10.1109/TITS.2020.3036984}}

@article{game_dec,
title={{Decision Making via Game Theory for Autonomous Vehicles in the Presence of a Moving Obstacle}},
author={Marina Vicini and Sercan Albut and Elvina Gindullina and Leonardo Badia},
journal={2022 IEEE International Conference on Communication, Networks and Satellite (COMNETSAT)},
year={2022},
pages={393-398},
url={https://api.semanticscholar.org/CorpusID:255417477}
}

@ARTICLE{game_lane,
author={Lopez, Victor G. and Lewis, Frank L. and Liu, Mushuang and Wan, Yan and Nageshrao, Subramanya and Filev, Dimitar},
journal={IEEE Transactions on Vehicular Technology}, 
title={{Game-Theoretic Lane-Changing Decision Making and Payoff Learning for Autonomous Vehicles}}, 
year={2022},
volume={71},
number={4},
pages={3609-3620},
doi={10.1109/TVT.2022.3148972}}

@INPROCEEDINGS{combin_drive,
author={Esterle, Klemens and Hart, Patrick and Bernhard, Julian and Knoll, Alois},
booktitle={2018 21st International Conference on Intelligent Transportation Systems (ITSC)}, 
title={{Spatiotemporal Motion Planning with Combinatorial Reasoning for Autonomous Driving}}, 
year={2018},
volume={},
number={},
pages={1053-1060},
doi={10.1109/ITSC.2018.8570003}}

@article{PadenCYYF16,
author       = {Brian Paden and
              Michal C{\'{a}}p and
              Sze Zheng Yong and
              Dmitry S. Yershov and
              Emilio Frazzoli},
title        = {{A Survey of Motion Planning and Control Techniques for Self-driving
              Urban Vehicles}},
journal      = {CoRR},
volume       = {abs/1604.07446},
year         = {2016},
url          = {http://arxiv.org/abs/1604.07446},
eprinttype    = {arXiv},
eprint       = {1604.07446},
bibsource    = {dblp computer science bibliography, https://dblp.org}
}

@INPROCEEDINGS{DeitsTedrake2015MIP,
author={Deits, Robin and Tedrake, Russ},
booktitle={2015 IEEE International Conference on Robotics and Automation (ICRA)}, 
title={{Efficient mixed-integer planning for UAVs in cluttered environments}}, 
year={2015},
volume={},
number={},
pages={42-49},
doi={10.1109/ICRA.2015.7138978}}

@inproceedings{EarlDAndrea2002MILP,
author    = {Earl, Michael G. and D'Andrea, Raffaello},
title     = {{Modeling and Control of a Multi-Agent System Using Mixed Integer Linear Programming}},
booktitle = {Proceedings of the IEEE Conference on Decision and Control (CDC)},
year      = {2002},
pages     = {107--111},
volume    = {1},
address   = {Las Vegas, NV, USA},
doi       = {10.1109/CDC.2002.1184476}
}

@inproceedings{RichardsHow2002MILP,
author    = {Richards, Arthur and How, Jonathan P.},
title     = {{Aircraft Trajectory Planning with Collision Avoidance Using Mixed Integer Linear Programming}},
booktitle = {Proceedings of the American Control Conference (ACC)},
year      = {2002},
volume    = {3},
pages     = {1936--1941},
address   = {Anchorage, AK, USA},
doi       = {10.1109/ACC.2002.1023918}
}

@ARTICLE{pot_game_drive,
author={Liu, Mushuang and Kolmanovsky, Ilya and Tseng, H. Eric and Huang, Suzhou and Filev, Dimitar and Girard, Anouck},
journal={IEEE Transactions on Intelligent Transportation Systems}, 
title={{Potential Game-Based Decision-Making for Autonomous Driving}}, 
year={2023},
volume={24},
number={8},
pages={8014-8027},
doi={10.1109/TITS.2023.3264665}}

@inproceedings{pot_ilqr,
title = {{Potential iLQR: A Potential-Minimizing Controller for Planning Multi-Agent Interactive Trajectories}},
author = "Talha Kavuncu and Ayberk Yaraneri and Negar Mehr",
year = "2021",
doi = "10.15607/RSS.2021.XVII.084",
language = "English (US)",
isbn = "9780992374778",
series = "Robotics: Science and Systems",
publisher = "Massachusetts Institute of Technology",
booktitle = "Robotics",
address = "United States",
}

@ARTICLE{pot_urb,
author={Zanardi, Alessandro and Mion, Enrico and Bruschetta, Mattia and Bolognani, Saverio and Censi, Andrea and Frazzoli, Emilio},
journal={IEEE Robotics and Automation Letters}, 
title={{Urban Driving Games With Lexicographic Preferences and Socially Efficient Nash Equilibria}}, 
year={2021},
volume={6},
number={3},
pages={4978-4985},
doi={10.1109/LRA.2021.3068657}}

@ARTICLE{out_MPC_game,
author={Wang, Mingyu and Wang, Zijian and Talbot, John and Gerdes, J. Christian and Schwager, Mac},
journal={IEEE Transactions on Robotics}, 
title={{Game-Theoretic Planning for Self-Driving Cars in Multivehicle Competitive Scenarios}}, 
year={2021},
volume={37},
number={4},
pages={1313-1325},
doi={10.1109/TRO.2020.3047521}}

@ARTICLE{Mixed_Pot_game,
author={Fabiani, Filippo and Grammatico, Sergio},
journal={IEEE Transactions on Intelligent Transportation Systems}, 
title={{Multi-Vehicle Automated Driving as a Generalized Mixed-Integer Potential Game}}, 
year={2020},
volume={21},
number={3},
pages={1064-1073},
doi={10.1109/TITS.2019.2901505}}

@article{GCS,
author = {Marcucci, Tobia and Umenberger, Jack and Parrilo, Pablo and Tedrake, Russ},
title = {{Shortest Paths in Graphs of Convex Sets}},
journal = {SIAM Journal on Optimization},
volume = {34},
number = {1},
pages = {507-532},
year = {2024},
doi = {10.1137/22M1523790},
URL = {https://doi.org/10.1137/22M1523790},
eprint = {https://doi.org/10.1137/22M1523790},
}

@article{ GCS_traj,
author = {Tobia Marcucci  and Mark Petersen  and David von Wrangel  and Russ Tedrake },
title = {{Motion planning around obstacles with convex optimization}},
journal = {Science Robotics},
volume = {8},
number = {84},
pages = {eadf7843},
year = {2023},
doi = {10.1126/scirobotics.adf7843},
URL = {https://www.science.org/doi/abs/10.1126/scirobotics.adf7843},
eprint = {https://www.science.org/doi/pdf/10.1126/scirobotics.adf7843}}

@article{GCS_SPP,
title={{Space-Time Graphs of Convex Sets for Multi-Robot Motion Planning}},
author={Tang, Jingtao and Mao, Zining and Yang, Lufan and Ma, Hang},
journal={arXiv preprint arXiv:2503.00583},
year={2025}
}

@INPROCEEDINGS{bakolas2021TaskAssignment,
author={Bakolas, Efstathios and Lee, Yoonjae},
booktitle={2021 American Control Conference (ACC)}, 
title={{Decentralized Game-Theoretic Control for Dynamic Task Allocation Problems for Multi-Agent Systems}}, 
year={2021},
volume={},
number={},
pages={3228-3233},
doi={10.23919/ACC50511.2021.9483030}}

@INPROCEEDINGS{lee2021RelayPursuit,
author={Lee, Yoonjae and Bakolas, Efstathios},
booktitle={2021 American Control Conference (ACC)}, 
title={{Relay Pursuit of an Evader by a Heterogeneous Group of Pursuers using Potential Games}}, 
year={2021},
volume={},
number={},
pages={3182-3187},
doi={10.23919/ACC50511.2021.9482912}}

@INPROCEEDINGS{dist_potential,
author={Williams, Zach and Chen, Jushan and Mehr, Negar},
booktitle={2023 IEEE International Conference on Robotics and Automation (ICRA)}, 
title={{Distributed Potential iLQR: Scalable Game-Theoretic Trajectory Planning for Multi-Agent Interactions}}, 
year={2023},
volume={},
number={},
pages={01-07},
doi={10.1109/ICRA48891.2023.10161176}}

@article{best_response_drone,
title = {{Multi-agent sensitivity enhanced iterative best response: A real-time game theoretic planner for drone racing in 3D environments}},
journal = {Robotics and Autonomous Systems},
volume = {125},
pages = {103410},
year = {2020},
issn = {0921-8890},
doi = {https://doi.org/10.1016/j.robot.2019.103410},
url = {https://www.sciencedirect.com/science/article/pii/S0921889019301939},
author = {Zijian Wang and Tim Taubner and Mac Schwager}}

@article{pot,
title = {{Potential Games}},
journal = {Games and Economic Behavior},
volume = {14},
number = {1},
pages = {124-143},
year = {1996},
issn = {0899-8256},
doi = {https://doi.org/10.1006/game.1996.0044},
url = {https://www.sciencedirect.com/science/article/pii/S0899825696900445},
author = {Dov Monderer and Lloyd S. Shapley}}

@misc{GCS_thesis,
author = {Tobia Marcucci},
howpublished = {Massachusetts Institute of Technology},
title = {{Graphs of Convex Sets with Applications to Optimal Control and Motion Planning}},
year = {2024}
}

@book{short_path,
title={{Linear Programming and Network Flows}},
author={Bazaraa, M.S. and Jarvis, J.J. and Sherali, H.D.},
isbn={9780470462720},
lccn={2009028769},
year={2009},
publisher={Wiley}}

@article{Kurtz2023GCSMotionPlanning,
author={Kurtz, Vince and Lin, Hai},
journal={IEEE Transactions on Robotics}, 
title={{Temporal Logic Motion Planning With Convex Optimization via Graphs of Convex Sets}}, 
year={2023},
volume={39},
number={5},
pages={3791-3804},
doi={10.1109/TRO.2023.3291463}}

@manual{mosek,
author = "MOSEK ApS",
title = {{The MOSEK Python Fusion API manual. Version 11.0.}},
year = 2025,
url = "https://docs.mosek.com/latest/pythonfusion/index.html"
}

@misc{gcsopt,
author       = {{Tobia Marcucci}},
title        = {{gcsopt: Graphs of Convex Sets optimization library}},
year         = {2026},
howpublished = {\url{https://github.com/TobiaMarcucci/gcsopt}},
note         = {GitHub repository, accessed March 2026}
}

@ARTICLE{Yan2023Decision,
author={Yan, Yongjun and Peng, Lin and Shen, Tong and Wang, Jinxiang and Pi, Dawei and Cao, Dongpu and Yin, Guodong},
journal={IEEE Transactions on Intelligent Vehicles}, 
title={{A Multi-Vehicle Game-Theoretic Framework for Decision Making and Planning of Autonomous Vehicles in Mixed Traffic}}, 
year={2023},
volume={8},
number={11},
pages={4572-4587},
doi={10.1109/TIV.2023.3321346}}

@article{Bhatt_2025,
title={{Strategic Decision-Making in Multiagent Domains: A Weighted Constrained Potential Dynamic Game Approach}},
volume={41},
ISSN={1941-0468},
url={http://dx.doi.org/10.1109/TRO.2025.3552325},
DOI={10.1109/tro.2025.3552325},
journal={IEEE Transactions on Robotics},
publisher={Institute of Electrical and Electronics Engineers (IEEE)},
author={Bhatt, Maulik and Jia, Yixuan and Mehr, Negar},
year={2025},
pages={2749–2764} }

@book{shoham2008multiagent,
title={{Multiagent systems: Algorithmic, game-theoretic, and logical foundations}},
author={Shoham, Yoav and Leyton-Brown, Kevin},
year={2008},
publisher={Cambridge University Press}
}

\section{Notation}\label{sec_notation}

Table~\ref{tab_notation} summarizes the notation used throughout the paper.

\begin{longtable}{@{}>{\raggedright\arraybackslash}p{0.23\textwidth} >{\raggedright\arraybackslash}p{0.73\textwidth}@{}}
\caption{Summary of notation.}\label{tab_notation}\\
\toprule
\textbf{Symbol} & \textbf{Description} \\
\midrule
\endfirsthead
\toprule
\textbf{Symbol} & \textbf{Description} \\
\midrule
\endhead
\midrule
\multicolumn{2}{r}{\textit{Continued on next page}} \\
\endfoot
\bottomrule
\endlastfoot

\multicolumn{2}{@{}l}{\textbf{General / macros}}\\
\addlinespace[2pt]
$\mcl{X}$ & Calligraphic font, $\eg$, $\mcl{N}$, $\mcl{L}$ (set-valued objects). \\
$\mbb{R}$ & Real numbers. \\
$\ubar{(\cdot)}$, $\obar{(\cdot)}$ & Lower/upper bounds ($\eg$, $\ubar{a}_i$, $\obar{a}_i$). \\
$\des$ & ``Desired'' superscript, $\eg$, $v_i^{\des}$, $\ell_i^{\des}$. \\
$\saf$ & ``Safety'' superscript, $\eg$, $d_i^{\saf}$. \\
$\delta_{i,j}$ & Kronecker delta (1 if $i=j$, 0 otherwise). \\

\addlinespace[6pt]
\multicolumn{2}{@{}l}{\textbf{Graphs of Convex Sets (GCS)}}\\
\addlinespace[2pt]
$\mcl{G}=(\mcl{V},\mcl{E})$ & Graph of convex sets with vertices $\mcl{V}$ and directed edges $\mcl{E}$. \\
$V\in\mcl{V}$ & A vertex of the GCS graph. \\
$e=(U,W)\in\mcl{E}$ & A directed edge from vertex $U$ to vertex $W$. \\
$\mcl{E}_V^{\mathrm{in}},\ \mcl{E}_V^{\mathrm{out}}$ & Incoming/outgoing edge sets at vertex $V$. \\
$\mcl{X}_V\subset\mbb{R}^n$ & Convex set associated with vertex $V$. \\
$\mcl{X}_e\subset\mbb{R}^n\times\mbb{R}^n$ & Convex set encoding feasibility/coupling across edge $e=(U,W)$. \\
$x_V\in\mcl{X}_V$ & Continuous decision variable at vertex $V$. \\
$z_e\in\{0,1\}$ & Binary edge-selection variable for edge $e$. \\
$y_V$ & Induced vertex-selection variable (vertex is on the selected path). \\
$c_e(\cdot)$ & Convex nonnegative edge cost. \\
$c_V(\cdot)$ & Optional convex nonnegative vertex cost. \\
$V_0,\ V_K$ & Source and target vertices for the shortest path problem. \\

\addlinespace[6pt]
\multicolumn{2}{@{}l}{\textbf{Game-theoretic model}}\\
\addlinespace[2pt]
$N$ & Number of agents. \\
$\mcl{N}=\{1,\dots,N\}$ & Set of agent indices. \\
$\theta_i$ & Strategy of agent $i$. \\
$\theta_{-i}$ & Joint strategy of all agents except $i$. \\
$\theta=(\theta_i,\theta_{-i})$ & Joint strategy profile. \\
$\Theta$ & Joint feasible set (captures coupled/shared constraints). \\
$\Theta_i(\theta_{-i})$ & Feasible strategy set of agent $i$ given $\theta_{-i}$. \\
$J_i(\theta_i,\theta_{-i})$ & Cost incurred by agent $i$ at joint strategy $(\theta_i,\theta_{-i})$. \\
$G$ & The generalized game defined by the coupled optimization problems. \\
$\theta^\star$ & A generalized Nash equilibrium (GNE) strategy profile. \\
$\Phi:\Theta\to\mbb{R}$ & Generalized potential function. \\
$\BR_i(\theta_{-i})$ & Best-response mapping of agent $i$. \\

\addlinespace[6pt]
\multicolumn{2}{@{}l}{\textbf{Iterative best response (IBR) and error metrics}}\\
\addlinespace[2pt]
$k$ & IBR iteration (sweep) index. \\
$\theta^k$ & Joint strategy at the start of sweep $k$. \\
$\theta^{k,i}$ & Intermediate joint strategy after agents $1,\dots,i$ updated in sweep $k$. \\
$\epsilon_i^k$ & Best-response suboptimality for agent $i$ at sweep $k$. \\
$\bar{\epsilon}$ & Uniform bound on best-response errors. \\
$r_i(\theta)$ & Regret of agent $i$ at profile $\theta$\\
$\epsilon$ & Tolerance used in termination criteria. \\

\addlinespace[6pt]
\multicolumn{2}{@{}l}{\textbf{Highway driving model}}\\
\addlinespace[2pt]
$\mcl{L}$ & Set of lane indices. \\
$\mcl{T}=\{0,\dots,T-1\}$ & Discrete planning horizon (time indices). \\
$\Delta t$ & Time step duration. \\
$s_i(t)$ & Longitudinal position of vehicle $i \in \mcl{N}$ at time $t$. \\
$v_i(t)$ & Longitudinal speed of vehicle $i \in \mcl{N}$ at time $t$. \\
$a_i(t)$ & Longitudinal acceleration of vehicle $i \in \mcl{N}$ at time $t$. \\
$z_i(t)\in\mcl{L}$ & Lane index of vehicle $i \in \mcl{N}$ at time $t$. \\
$\mcl{B}=\{-1,0,1\}$ & Blinker command set (right, stay, left). \\
$b_i(t)\in\mcl{B}$ & Blinker (lane-change command) of vehicle $i \in \mcl{N}$ at time $t$. \\
$\ubar{a}_i,\ \obar{a}_i$ & Acceleration bounds for vehicle $i \in \mcl{N}$. \\
$\ubar{v}_i,\ \obar{v}_i$ & Speed bounds for vehicle $i \in \mcl{N}$. \\
$\ubar{s},\ \obar{s}$ & Road segment bounds in longitudinal coordinate. \\
$d_{i,j}(t)=s_j(t)-s_i(t)$ & Longitudinal separation between vehicles $i$ and $j$. \\
$z_{i,j}(t)=z_j(t)-z_i(t)$ & Relative lane offset between vehicles $i$ and $j$. \\
$d_i^{\saf}$ & Safety distance (longitudinal gap) for vehicle $i \in \mcl{N}$. \\

\addlinespace[6pt]
\multicolumn{2}{@{}l}{\textbf{Single-agent GCS construction for vehicle $i \in \mcl{N}$}}\\
\addlinespace[2pt]
$\mcl{G}_i=(\mcl{V}_i,\mcl{E}_i)$ & Vehicle-$i$ GCS graph used in its best-response subproblem. \\
$\ell\in\mcl{L}$ & Lane index in the construction. \\
$g$ & Safe-gap index within a lane at a given time. \\
$\mcl{K}_i^{j}(\ell,t)$ & Unsafe longitudinal interval around vehicle $j$ in lane $\ell$ at time $t$. \\
$\mcl{F}_i(\ell,t)$ & Union of unsafe intervals in lane $\ell$ at time $t$. \\
$\mcl{S}_i(\ell,t)$ & Collision-free longitudinal set in lane $\ell$ at time $t$. \\
$\mcl{H}_i(\ell,t)$ & Finite index set of collision-free gaps in lane $\ell$ at time $t$. \\
$\mcl{I}_i^{\ell,g}(t)$ & The $g$-th collision-free longitudinal interval in lane $\ell$ at time $t$. \\
$V_i(t,\ell,g)$ & Vertex corresponding to $(t,\ell,g)$ in $\mcl{G}_i$. \\
$x_i(t)=\begin{bmatrix}s_i(t)\\ v_i(t)\end{bmatrix}$ & Continuous state used in vertices of $\mcl{G}_i$. \\
$e=\bigl(V_i(t,\ell,g),V_i(t+1,\ell',g')\bigr)$ & Directed edge representing a feasible transition between time-adjacent vertices. \\

\addlinespace[6pt]
\multicolumn{2}{@{}l}{\textbf{Costs / preferences}}\\
\addlinespace[2pt]
$v_i^{\des}$ & Desired speed of vehicle $i \in \mcl{N}$. \\
$\ell_i^{\des}$ & Desired lane of vehicle $i \in \mcl{N}$. \\
$w_{i,v}, w_{i,\ell}, w_{i,a}, w_{i,b}$ & Cost weights for speed tracking, lane preference, acceleration effort, and blinker usage. \\

\end{longtable}

\end{document}